\documentclass[a4paper,11pt]{article}%amsart}
\usepackage[utf8x]{inputenc}
\usepackage{a4wide}
\usepackage{amssymb}
\usepackage{amsmath}
\usepackage{amsthm}
\usepackage{latexsym}
\usepackage[mathscr]{euscript}
\usepackage{hyperref}

\theoremstyle{definition}
\newtheorem{definition}{Definition}[section]

\theoremstyle{plain}
\newtheorem{Theorem}[definition]{Theorem}
\newtheorem{Proposition}[definition]{Proposition}
\newtheorem{Lemma}[definition]{Lemma}
\newtheorem{Corollary}[definition]{Corollary}

\theoremstyle{remark}
\newtheorem{remark}[definition]{Remark}

\newcommand{\R}{\mathbb R}  
\newcommand{\N}{\mathbb N}

\newcommand{\grad}{\mathrm{grad}}

\newcommand{\eps}{\varepsilon}

\usepackage[color]{changebar}
\cbcolor{blue}
\newcommand{\Ric}{\mathrm{Ric}}
\newcommand{\gec}{{\check g_\eps}}
\newcommand{\geck}{{\check g_{\eps_j}}}
\newcommand{\gameck}{{\gamma_{\eps_j}}}
\newcommand{\comp}{\Subset}
\newcommand{\X}{\mathfrak{X}}

\newcommand{\sse}{\subseteq}
\newcommand{\conv}{\mathbf{k}}
\newcommand{\edge}{\mathrm{edge}}

\title{Hawking's singularity theorem for $C^{1,1}$-metrics}

\author{Michael Kunzinger\footnote{University of Vienna, Faculty of Mathematics, 
michael.kunzinger@univie.ac.at, roland.steinbauer@univie.ac.at, milena.stojkovic@live.com}, 
Roland Steinbauer\footnotemark[\value{footnote}], Milena Stojkovi\'c\footnotemark[\value{footnote}],
James A.\ Vickers\footnote{University of Southampton, School of Mathematics, J.A.Vickers@maths.soton.ac.uk}}

\begin{document}

\date{Jannaury 12, 2015}

%\date{Received: date /Accepted: date}

\maketitle

\begin{abstract}
We provide a detailed proof of Hawking's singularity theorem in the regularity 
class $C^{1,1}$, i.e., for spacetime metrics possessing locally Lipschitz
continuous first derivatives. The proof uses recent results in
$C^{1,1}$-causality theory and is based on regularisation techniques adapted
to the causal structure.

\vskip 1em

\medskip


\end{abstract}

\section{Introduction}

In the early years of General Relativity it was known that there
existed solutions of the Einstein field equations which had singular
behaviour of various kinds. However, the prevailing view was that these
singularities were the result of the high degree of symmetry or were
unphysical in some way.  This position changed considerably with the
work of Penrose who showed in his 1965 paper \cite{Pen} that deviations from
spherical symmetry could not prevent gravitational collapse. This
paper not only introduced the concept of closed trapped surface, but
used the notion of geodesic incompleteness to characterise a singular
spacetime. 
%Further\-more as described in a recent review article
%\cite{SenGar} this paper inspired many of the ideas of mathematical
%relativity such as extensions of spacetime and causality theory.

Shortly afterwards Hawking realised
that by considering a closed trapped surface to the past one could
show that an approximately homogeneous and isotropic cosmological
solution must have an initial singularity. There quickly followed a
series of papers by Hawking, Penrose, Ellis, Geroch and others which
led to the development of modern singularity theorems, one of the greatest
achievements within general relativity. (See the recent review paper
\cite{SenGar} for details.) The resulting theorems all had the same general
framework described by Senovilla in \cite{Seno1} as a ``pattern singularity
theorem''. 
\medskip

\vbox{
\noindent
{\bf Pattern Singularity Theorem.} 
\emph{If a spacetime with a $C^2$-metric satisfies
\begin{enumerate}%[(i),topsep=1pt,itemsep=-1pt]
\item[(i)] a condition on the curvature
\item[(ii)] a causality condition
\item[(iii)] an appropriate initial and/or boundary condition
\end{enumerate}
then it contains endless but incomplete causal geodesics.}
}
\medskip

%Note that despite the power of these theorems they simply show that
%the spacetime is causally geodesically incomplete and say little about
%the nature of the singularity. 
Despite their power and glory the singularity theorems have a weak
point, which is their conclusion. In fact, they simply show causal
geodesic incompleteness of the spacetime but say little about
the nature of the singularity. In particular, they 
do not say that the curvature blows up (see, however \cite{Cl82, Clarke} as well as \cite[Sec.\ 5.1.5]{SenGar}
and the references therein)
and it could be that the singularity is simply a result of the
differentiability dropping below $C^2$. 
In the case that the regularity of the metric simply dropped to $C^{1,1}$
(also denoted by $C^{2-}$, the first derivatives of the metric being locally
Lipschitz continuous) the theorems would predict the curvature to become
discontinuous rather than unbounded. Recall that indeed the connection of a
$C^{1,1}$-metric is locally Lipschitz and hence by Rademacher's
theorem differentiable almost everywhere with locally bounded curvature.
From the viewpoint of physics such a
situation would hardly be regarded as `singular' as it corresponds, via the
field equations, to a finite jump in the matter variables. There are many
physically realistic systems of that type, such as the Oppenheimer-Snyder model of
a collapsing star \cite{OppSny}, to give a classical example, and
general matched spacetimes, see e.g.\ \cite{L,MaSeno}.

Also from the point of view of the singularity theorems themselves  
the natural differentiability class is $C^{1,1}$. Indeed this is the
minimal condition which ensures existence and uniqueness of solutions of the geodesic
equation, which is essential to the statement of the theorems.
Moreover, as already pointed out in \cite[Sec.\ 8.4]{HE}, in the context of a
$C^{1,1}$-singularity theorem a further dropping of the regularity would 
result in spacetimes where generically the curvature diverges and in addition
there are problems with the uniqueness of causal geodesics and hence the worldlines of
physical observers. Such a situation could be interpreted as physically 
`singular' with much better reason than the corresponding $C^2$-situation
discussed above.

All this provides a strong motivation for trying to prove the
singularity theorems in the regularity class $C^{1,1}$. In \cite[Sec.\
8.4]{HE} Hawking and Ellis discuss the nature of the singularities
predicted by the singularity theorems and go on to outline a proof of
Hawking's singularity theorem based on an approximation of the
$C^{1,1}$-metric by a 1-parameter family of smooth metrics.  However
the $C^2$-differentiability assumption plays a key role in many places
in the singularity theorems and it is not obvious that these can all
be dealt with without having further information about the nature of
the approximation. Indeed much of standard causality theory assumes
that the metric is smooth or at least $C^2$, see e.g.\
\cite{HE,MS,Clarke,Seno1,Seno2,Chrusciel_causality,SW} for a review of
various approaches to causal structures and discussions of the
regularity assumptions. Senovilla in \cite[Sec.\ 6.1]{Seno1} lists
those places where the $C^2$-assumption explicitly enters the proofs
of the singularity theorems, indicating the number of technical
difficulties a proof in the $C^{1,1}$-case would have to overcome.
Indeed, to our knowledge, the only results that are available in
$C^{1,1}$-singularity theory are very limited \cite{Ca70,Cl82,Clarke}
or restricted to special situations \cite{MMS} and we think it is
fair to say (cf.\ \cite{Seno1}) that the issue of regularity in the
singularity theorems is often ignored despite its mathematical and
physical relevance.
\medskip

Motivated by the phyiscal arguments given above and recent advances in
the regularity required for the initial value problem (see e.g.\ \cite{KR05})
there has been an increased interest in causality theory of spacetimes
of low regularity. Chrusciel and Grant in \cite{CG} adopted a
regularisation approach which is adapted to the causal structure: a
given metric of low regularity is approximated by two nets of smooth
metrics ${\check g}_\epsilon$ and ${\hat g}_\epsilon$ whose light
cones sandwich those of $g$. They established some fundamental
elements of causality theory in low regularity such as the existence
of smooth time functions on domains of dependence even for continuous
metrics, see also \cite{FS11}. However, they also revealed a dramatic
failure of fundamental results of smooth causality if the regularity
was below $C^{1,1}$.  In particular, they demonstrated the existence
of `bubbling metrics' (of regularity $C^{0,\alpha}$, for any $\alpha\in (0,1)$), 
whose light-cones have nonempty interior,
thereby nicely complementing classical examples by Hartman and Wintner
\cite{Hartman,HW} which demonstrate the failure of convexity
properties in the Riemannian case.

One of the key technical tools
employed in causality theory is the exponential map and the existence
of totally normal neighbourhoods which allow one to relate the causal
structure of Minkowski space to that of the manifold in any given point
(see Theorem \ref{lcb}).
%ensure that locally the causal
%structure is that of Minkowski space. 
Classical results for $C^{1,1}$-metrics 
only show that the exponential map is a local homeomorphism \cite{Whitehead},
which
is insufficient to establish the required results. Recently, however, it
has been shown that $\exp$ is a bi-Lipschitz homeomorphism. Using a
careful analysis of the corresponding ODE problem based on
Picard-Lindel\"of approximations, as well as an inverse function
theorem for Lipschitz maps, Minguzzi \cite{M} established the fact that
$\exp$ is a bi-Lipschitz homeomorphism \cite[Th.\ 1.11]{M} and
used this to derive many standard results in causality theory. 
Also the present authors in \cite[Th.\ 2.1]{KSS} and \cite{KSSV}
established similar results by extending the refined regularisation
methods of \cite{CG} and combining them with methods from comparison
geometry \cite{CleF}.  \medskip

Given that finally the key elements of causality theory are
in place, now is the time to approach the singularity theorems
for $C^{1,1}$-metrics. Indeed, in this work we will show that the tools now
available allow one to prove singularity theorems with $C^{1,1}$-regularity and we
illustrate this by providing a rigorous proof of Hawking's theorem in the
$C^{1,1}$-regularity class. To be precise we establish the following result:

%%%%%%%%%%%%%%%%%%%%%%%%%%%%%%%%%%%%%%%%%%%%%%%%%%%%%%%%%%%%%%%%%%%%%%%%%%%%%%%%%%%%%%%%%%%%%%%%%%%%
\begin{Theorem}\label{hawking} Let $(M,g)$ be a $C^{1,1}$-spacetime. Assume
\begin{itemize}
\item[(i)] For any smooth timelike local vector field $X$,
$\Ric(X,X)\ge 0$.
\item[(ii)] There exists a compact spacelike hypersurface $S$ in $M$.
\item[(iii)] The future convergence $\conv$ of $S$ is everywhere strictly positive. 
\end{itemize}
Then $M$ is future timelike geodesically incomplete.
\end{Theorem}
\begin{remark}\ 
\begin{itemize}
\item[(i)]
For the definition of a $C^{1,1}$-spacetime, see Section \ref{prelim}.
Since $g$ is $C^{1,1}$, its Ricci-tensor is of regularity $L^\infty$. In particular,
it is in general only defined almost everywhere. For this reason, we have cast
the curvature condition (i) in the above form. For any smooth vector field $X$ defined
on an open set $U\sse M$, $\Ric(X,X)\in L^\infty(U)$, so $\Ric(X,X)\ge 0$ means
that $\Ric_p(X(p),X(p))\ge 0$ for almost all $p\in U$. Since any timelike $X\in T_pM$
can be extended to a smooth timelike vector field in a neighborhood of $p$, (i) is equivalent
to the usual pointwise condition ($\Ric(X,X)\ge 0$ for any timelike $X\in TM$) if
the metric is $C^2$.

\item[(ii)] Concerning (iii) in the Theorem, our conventions (in accordance with \cite{ON83}) 
are that $\conv=\mathop{tr}S_U/(n-1)$ and
$S_U(V)= -\nabla_VU$ is the shape operator of $S$, where $U$ is the
future pointing unit normal, $\nabla$ denotes the connection
on $M$ and $V$ is any vector field on the embedding $S\hookrightarrow
M$.  

\item[(iii)] In the physics literature, the negative of what we call the future convergence is often
denoted as the expansion of $S$. 

\item[(iv)] Finally, we note that an analogous result for past timelike incompleteness holds
if the convergence in (iii) of the Theorem is supposed to be everywhere strictly negative.
\end{itemize}
\end{remark}
\medskip

In proving this theorem we will follow the basic strategy outlined in
\cite[Sec.\ 8.4]{HE}. However, in our proof we will make extensive use of the
recent results of $C^{1,1}$-causality theory. An important feature of
this paper is that we carefully collect all the results from $C^{1,1}$-causality
theory that are required for the proof of the above theorem
and show how they can be obtained from \cite{CG,M,KSSV}. 
 In addition, in section 4 we make crucial use of causal
regularisation techniques to show the existence of maximising curves. We
therefore need to establish the existence of an approximating family
of smooth metrics which satisfy (a weakened form of) the Ricci
convergence condition while at the same time controlling the light cones. 

The plan of the paper is as follows. In section 2 we fix the
definitions and notation we will use in the rest of the
paper. In section 3 we introduce the causal regularisation
techniques and establish the required estimates for the curvature. In
section 4 we make use of the causal regularisation together with some
key results from $C^{1,1}$-causality theory to establish the existence
of maximal curves. Finally in section 5 we prove the main result
following the basic layout of the proof of \cite[Th.\ 14.55B]{ON83}. In
the Appendix we collect together all the results from causality theory
that are required and show how they are proved in the $C^{1,1}$-case.

\section{Preliminaries}\label{prelim}

In this section we fix key notions to be used throughout this paper.
We assume all manifolds to be of class $C^\infty$ (as well as second
countable), and only lower the regularity of the metric. This is no
loss of generality since any $C^k$-manifold $M$ with $k\ge 1$
possesses a unique $C^\infty$-structure that is $C^k$-compatible with
the given $C^k$-structure on $M$ (see \cite[Th.\ 2.9]{Hirsch}). Most
of the time (and unless explicitly stated otherwise) we will deal with
a $C^{1,1}$-spacetime $(M,g)$, by which we mean a smooth manifold $M$
of dimension $n$ endowed with a time-oriented Lorentzian metric $g$ of
signature $(-+\dots+)$ possessing locally Lipschitz continuous first
derivatives and with the time orientation given by a continuous timelike
vector field. If $K$ is a compact set in $M$ we write $K\comp M$.
Following \cite{ON83}, we define the curvature tensor to be given by
$R(X,Y)Z=\nabla_{[X,Y]}Z - [\nabla_X,\nabla_Y]Z$. This convention
differs by a sign from that of \cite{HE}. We then define the Ricci
tensor by $R_{ab}=R^c{}_{abc}$ (which again differs by a sign from that in \cite{HE} where $R_{ab}=R^c{}_{acb}$, so overall the two definitions of Ricci curvature agree).

There are minor variations in the basic definitions used
in causality theory by various authors and this section serves to
specify the ones we will be using and relate them to those used
elsewhere.  Our notation for causal structures will basically follow
\cite{ON83} although following \cite{Chrusciel_causality,KSSV} we will
base all causality notions on locally Lipschitz curves. We note that
in most of the standard literature on causality theory, in particular
in \cite{HE,ON83}, the
%this definition differs from that in \cite{M}, where the 
corresponding curves are required to be (piecewise) $C^1$. However, as
is shown in \cite[Th.\ 1.27]{M}, \cite[Cor.\ 3.1]{KSSV}, this does not
affect the definition of (causal or chronological) pasts and futures.
Any locally Lipschitz curve $c$ is differentiable almost everywhere
(by Rademacher's theorem) and its derivative is locally bounded.
We call $c$ timelike, causal, spacelike
or null, if $c'(t)$ has the corresponding property almost everywhere.
Based on these notions we define the relative chronological future
$I^+(p,A)$ and causal future $J^+(p,A)$ of $p$ in $A\subseteq M$
literally as in the smooth case (see \cite[Def.\ 3.1]{KSSV}
\cite[Sec.\ 2.4]{Chrusciel_causality}). For $B\subseteq A$ we set $I^+(B,A) :=
\bigcup_{p\in B} I^+(p,A)$ and
analogously for $J^+(B,A)$. Moreover, we set $I^+(p):=I^+(p,M)$. The
same conventions apply to the respective past sets where the $+$ is
replaced by a $-$. For $p$, $q\in M$ we write $p < q$, respectively $p\ll q$, 
if there is a future directed causal,
respectively timelike, curve from $p$ to $q$. By $p\le q$ 
we mean $p = q$ or $p < q$. We denote the time separation (distance)
between two points $p,q\in M$ and between $A,B\subseteq M$ with
respect to some Lorentzian metric $g'$ by $d_{g'}(p,q)$ and
$d_{g'}(A,B)$, respectively (cf.\ \cite[Def.\ 14.15]{ON83}). 
We call a $C^{1,1}$-spacetime $(M,g)$ globally hyperbolic if it is strongly causal
and $J(p,q):=J^+(p)\cap J^-(q)$ is compact for all $p,\,q\in M$. Finally,
for an achronal set $S$, the future Cauchy development of $S$ is the
set $D^+(S)$ of all points $p\in M$ with the property that every past
inextendible causal curve through $p$ meets $S$. Then
$H^+(S):=\overline{D^+(S)}\setminus I^-(D^+(S))$ is its future Cauchy
horizon. Note that both, Cauchy development and Cauchy horizon, are
defined with locally Lipschitz causal curves (contrary to
\cite{HE,ON83}).  That this does not affect our considerations is
shown in Lemma \ref{lipdevelopment}.  A Cauchy hypersurface is a
subset $S$ of $M$ which every inextendible timelike curve intersects
exactly once, see \cite[Def.\ 14.28]{ON83}.  In the smooth case,
for spacelike hypersurfaces 
this definition of a Cauchy hypersurface is equivalent to the one in \cite{HE},
and this remains true in the $C^{1,1}$-case, cf.\ Proposition
\ref{cauchysurface}.

Now let $S$ be a spacelike hypersurface in $M$ with a Lorentzian
metric $g$.  By $N(S)$ we denote the set of vectors perpendicular to
$S$ with respect to the metric $g$ and by $(N(S), \pi)$ the normal
bundle of $S$ in $M$, where $\pi: N(S)\rightarrow S$ is the map
carrying each vector $v\in T_p(S)^\perp$ to $p\in S$. We will
distinguish normal bundles stemming from metrics $g_{\eps}$ by writing
$(N_{g_{\eps}}(S), \pi_{g_{\eps}})$ and for brevity we will drop this
subscript for the $C^{1,1}$-metric $g$ itself.  The exponential map
with respect to the metric $g$ generalises in the following way: the
normal exponential map
\begin{equation*}
\exp^\perp: N(S)\rightarrow M 
\end{equation*}
assigns to a vector $v\in N(S)$ the point $c_v(1)$ in $M$, where $c_v$ is the
geodesic with initial data $v$. Thus 
$\exp^\perp$ carries radial lines in $T_pS$ to geodesics of $M$ that are normal to $S$ at $p$.
Again, in order to distinguish the normal exponential maps w.r.t.\ metrics
$g_{\eps}$, we write $\exp_{g_{\eps}}^\perp$. As was shown in \cite[Th.\
1.39]{M}, $N(S)$ is a Lipschitz bundle and $\exp^\perp$ is a bi-Lipschitz
homeomorphism from a neighbourhood of the zero section in $N(S)$ onto a
neighbourhood of $S$ (cf.\ Th.\ \ref{subnormalnbhds} below).

\section{Regularisation techniques}

While the relevance of regularisation techniques to the problem at hand was
already clearly pointed out in \cite[Sec.\ 8.4]{HE} we 
shall see at several places below that a straightforward regularisation via
convolution in charts (as in \cite[Sec.\ 8.4]{HE}) is insufficient to actually
reach the desired conclusions. Rather, techniques adapted to the causal 
structure as introduced in \cite{CG} will be needed.
This remark, in particular, applies to
the results on the existence of maximising curves (Lemma \ref{lengthcomp} and
Proposition \ref{prop2.2}) below as well as to the proof of the main result in
Section \ref{mainproof}.

Recall from \cite[Sec.\ 3.8.2]{MS}, \cite[Sec.\ 1.2]{CG} that for two Lorentzian metrics $g$,
$h$, we say that $h$ has \emph{strictly wider light cones} than $g$, denoted by 
\begin{equation}
 g\prec h, \text{ if for any tangent vector } X\not=0,\ g(X,X)\le 0 \text{ implies that } h(X,X)<0.
\end{equation}
The key result now is \cite[Prop.\ 1.2]{CG}, which we give here in the slightly refined
version of  \cite[Prop.\ 2.5]{KSSV}: 
\begin{Proposition}\label{CGapprox} Let $(M,g)$ be a spacetime with a
continuous Lorentzian metric, and $h$ some smooth
background Riemannian metric on $M$. Then for any $\eps>0$, there exist smooth
Lorentzian metrics $\check g_\eps$ and $\hat g_\eps$ on $M$ such that $\check g_\eps
\prec g \prec \hat g_\eps$ and $d_h(\check g_\eps,g) + d_h(\hat g_\eps,g)<\eps$,
where  
\begin{equation}
d_h(g_1,g_2) := \sup_{p\in M,0\not=X,Y\in T_pM} \frac{|g_1(X,Y)-g_2(X,Y)|}{\|X\|_h
\|Y\|_h}.
\end{equation}
Moreover, $\hat g_\eps$ and $\check g_\eps$ depend smoothly on $\eps$, and if
$g\in C^{1,1}$ then letting $g_\eps$ be either $\check g_\eps$ or $\hat g_\eps$,
we additionally have %satisfy 
\begin{itemize}
 \item[(i)] $g_\eps$ converges to $g$ in the $C^1$-topology as $\eps\to 0$, and
 \item[(ii)] the second derivatives of $g_\eps$ are bounded, uniformly in $\eps$, on compact sets.
 \end{itemize}
%(i) and (ii) from Construction \ref{approxrem}.
\end{Proposition}

One essential assumption in the singularity Theorem~\ref{hawking} is the curvature
condition (i) for the $C^{1,1}$-metric $g$. We now derive from it a (weaker)
curvature condition for any approximating sequence $\gec$ as in Proposition
\ref{CGapprox}, 
which is vital in our proof of the main theorem. This should be compared to 
condition (4) on p.\ 285 of \cite{HE}.
%%%%%%%%%%%%%%%%%%%%%%%%%%%%%%%%%%%%%%%%%%%%%%%%%%%%%%%%%%%%%%%%%%%%%%%%%%%%%%%%%%%%%%%%%%%%%%%%%%%%
\begin{Lemma}\label{(4)} Let $M$ be a smooth manifold with a $C^{1,1}$-Lorentzian metric $g$
and smooth background Riemannian metric $h$.
Let $K$ be a compact subset of $M$  
and suppose that $\Ric(X,X)\ge 0$ for every $g$-timelike smooth local vector field $X$.
Then 
%given any $C>0$, any $\delta>0$ and any $\kappa<0$ there exists some $\eps_0>0$ such that
%for any $X\in TM|_K$
\begin{equation}\label{suffest}
\begin{split}
&\forall C>0\ \forall \delta>0\ \forall \kappa<0\ \exists \eps_0>0\ \forall \eps<\eps_0\ \forall X\in TM|_K \\ 
&\text{ with }\   g(X,X)\le \kappa \ \text{ and } \|X\|_h \leq C \text{ we have } \ \Ric_\eps(X,X) > -\delta.
\end{split}
\end{equation}
Here $\Ric_\eps$ is the Ricci-tensor corresponding to a metric $\gec$ as in Proposition \ref{CGapprox}.
\end{Lemma}
%%%%%%%%%%%%%%%%%%%%%%%%%%%%%%%%%%%%%%%%%%%%%%%%%%%%%%%%%%%%%%%%%%%%%%%%%%%%%%%%%%%%%%%%%%%%%%%%%%%%
\begin{proof} Let us first briefly recall the notations from the proof of \cite[Prop.\ 2.5]{KSSV}:
Let $(U_i,\psi_i)$ ($i\in \N$) be a countable and locally finite collection of
relatively compact charts of $M$ and denote by $(\zeta_i)_i$ a
subordinate partition of unity with $\mathrm{supp}(\zeta_i)\Subset U_i$ (i.e., 
$\mathrm{supp}(\zeta_i)$ is a compact subset of $U_i$) for all $i$. Moreover, 
choose a family of cut-off functions $\chi_i\in\mathscr{D}(U_i)$ with
$\chi_i\equiv 1$ on a
neighbourhood of $\mathrm{supp}(\zeta_i)$. Finally, let $\rho\in
\mathscr{D}(\R^{n})$ be a non-negative test function with unit integral and define the
standard mollifier $\rho_{\eps}(x):=\eps^{-n}\rho\left (\frac{x}{\eps}\right)$
($\eps>0$). By $f_*$ (resp.\ $f^*$) we denote push-forward
(resp.\ pullback) under a smooth map $f$. 
%the following formula defines a family 
%$(g_\eps)_\eps$ of smooth sections of $T^0_2(M)$ 
It then follows from (2.2) in the proof of \cite[Prop.\ 2.5]{KSSV} that 
\begin{equation}
\gec - \sum\limits_i\chi_i\,\psi_i^*\Big(\big(\psi_{i\,*} (\zeta_i\,
g)\big)*\rho_{\eta(\lambda_i(\eps),i)}\Big) \to 0 \text{ in } C^2(M).
\end{equation}
Since $\eta(\lambda_i(\eps),i)\to 0$ as $\eps\to 0$ and %$\gec(X,X)\to g(X,X)$ uniformly on 
$\{X\in TM|_K\mid \|X\|_h \leq C\}$ is compact, %as $\eps\to 0$ 
we conclude that in order to establish the result it will suffice to assume that $M=\R^n$,
$\|\,.\,\|_h = \|\,.\,\|$ is the Euclidean norm, to replace $\gec$ by $g_\eps:=g*\rho_\eps$ 
(component-wise convolution), and prove \eqref{suffest} for $\Ric_\eps$ calculated from $g_\eps$.
%\begin{equation}\label{asdf}
%\begin{aligned}
%\forall C>0\ \forall \delta>0\ \exists \eps_0>0\ %\exists \delta_0>0\ 
%\forall \eps<\eps_0\ &\forall X\in TM|_K 
%\ \text{ with }\ g(X,X)\le \delta_0 \ \\ \text{ and } \|X\|_h \leq C:\ 
%& \Ric_\eps(X,X) > -\delta.
%\end{aligned}
%\end{equation}
%
%Next we observe that, by continuity, \eqref{asdf} follows once we are able to show that
%on any compact set $K$ and for any $C>0$ we have:
%\begin{equation}\label{lkjh}
%\begin{aligned}
%\forall \delta>0\ \exists \eps_0>0\  \forall \eps<\eps_0\  &\forall X\in TM|_K 
%\ \text{ with }\ g(X,X)\le 0 \ \text{ and } \|X\|_h \leq C:\\
% & \Ric_\eps(X,X) > -\delta.
%\end{aligned}
%\end{equation}

We first claim that 
\begin{equation}\label{www}
R_{\eps jk} - R_{jk}*\rho_\eps \to 0 \ \text{ uniformly on compact sets}.
\end{equation}

We have $R_{jk} = \partial_{x^i}\Gamma^i_{kj} - \partial_{x^k}\Gamma^i_{ij} + \Gamma^i_{ij}\Gamma^m_{kj} 
- \Gamma^i_{km}\Gamma^m_{ij}$. In this expression, all terms involving at most first derivatives of
$g$ are uniform limits of the corresponding terms in $R_{\eps jk}$, while the remaining terms are of the
form $g^{im}a_{ijkm}$, where $a_{ijkm}$ consists of second derivatives of $g$. These observations imply
that \eqref{www} will follow from the following mild variant of the Friedrichs lemma:

{\em Claim:} Let $f\in \mathcal{C}^0(\R^n)$, $a\in L^\infty_{\mathrm{loc}}(\R^n)$. Then $(f\cdot a)*\rho_\eps - (f*\rho_\eps)\cdot(a*\rho_\eps)
\to 0$ locally uniformly.

In fact,
\begin{equation}
\begin{aligned}
\Big((f\cdot a)*\rho_\eps - (f*\rho_\eps)\cdot(a*\rho_\eps)\Big)(x)
= \int \Big(f(y)-\big(f*\rho_\eps\big)(x)\Big)a(y)\rho_\eps(x-y)\,dy\\
= \int \Big(f(y)-f(x)\Big)a(y)\rho_\eps(x-y)\,dy + \int
\Big(f(x)-\big(f*\rho_\eps\big)(x)\Big)a(y)\rho_\eps(x-y)\,dy,
\end{aligned}
\end{equation}
so for any $L\comp \R^n$ we obtain
\begin{equation}
\begin{aligned}
\sup_{x\in L}|(f\cdot a)*\rho_\eps - (f*\rho_\eps)\cdot(a*\rho_\eps)|(x)
&\le (\max_{\substack{x\in L\\ |x-y|\le\eps}} |f(y)-f(x)|)
\cdot \sup_{d(y,L)\le \eps}|a(y)| \\
&+(\sup_{x\in L}|f(x)-f_\eps(x)|) \cdot \sup_{d(y,L)\le \eps}|a(y)| \to 0 
\end{aligned}
\end{equation}
as $\eps\to 0$, so \eqref{www} follows.

Since $g$ is uniformly continuous on $K$ there exists some $r>0$ such that for 
any $p,\, x\in K$ with $\|p-x\|<r$ and any $X\in \R^n$ with $\|X\|\le C$ we have $|g_p(X,X)-g_x(X,X)|< -\kappa$.
Now % fix some compact neighbourhood $L$ of $K$, 
let $p\in K$ and let $X\in \R^n$ be any vector such that 
$g_p(X,X)\le\kappa$ and $\|X\|\le C$. Then on the open ball $B_r(p)$ the constant vector
field $x\mapsto X$ (i.e., the map that assigns to each $x\in B_r(p)$ this same vector $X\in \R^n$), 
which we again denote by $X$, is $g$-timelike.

%Extend the function $x\mapsto R_{jk}(x)X^jX^k$ to all of $\R^n$ by setting it equal to $0$ outside of $U$.
Let 
\begin{equation}
\tilde R_{jk}(x) := \left\{
\begin{array}{rl}
	R_{jk}(x) & \text{ for } x\in B_r(p)\\
	0 & \text{ otherwise}
\end{array}\right.
\end{equation}

By our assumption and the fact that $\rho\ge 0$ we then have $(\tilde R_{jk}X^jX^k)*\rho_\eps\ge 0$ on $\R^n$.
Moreover, for $\eps<r$ it follows that $(R_{jk}*\rho_\eps)(p) = (\tilde R_{jk}*\rho_\eps)(p)$.
%Moreover, there exists some $\tilde \eps_0>0$ (depending only on $L$ \todo{check})
%such that for $\eps\le \tilde\eps_0$ and any $x\in U$ we have
Thus for such $\eps$ we have %for any $x\in U$ we obtain
\begin{equation}
\begin{aligned}
|R_{\eps jk}(p)X^jX^k - ((\tilde R_{jk}X^jX^k)*\rho_\eps)(p)| &= |(R_{\eps jk}(p) - (R_{jk}*\rho_\eps)(p))X^jX^k| \\
&\le C^2 \sup_{x\in K} |R_{\eps jk}(x) - R_{jk}*\rho_\eps(x)|.\end{aligned}
\end{equation}
Using \eqref{www} we conclude from this estimate that, given any $\delta>0$ we may choose $\eps_0$ 
such that for all $\eps<\eps_0$, all $p\in K$ and all vectors $X$ with $g_p(X,X)\le \kappa$ and $\|X\|\le C$ we have
$R_{\eps jk}(p)X^jX^k>-\delta$, which is \eqref{suffest}.
%It therefore suffices to show that $R_{\eps jk}(x)X^jX^k - (R_{jk}X^jX^k)*\rho_\eps(x)\to 0$
%uniformly in $x\in K$ and in $X$ as in \eqref{suffest}, which in turn is implied if we can show 
\end{proof}

\section{Existence of maximal curves}\label{sec:4}

The next key step in proving the main result is to secure the existence of geodesics maximising the distance
to a spacelike hypersurface. To prove this statement we will employ a net $\gec$ ($\eps>0$) of 
smooth Lorentzian metrics whose lightcones approximate those of $g$ from the inside as in Prop.\ \ref{CGapprox}.
We first need some auxiliary results.

\begin{Lemma}\label{lhbound} Let $(M,g)$ be a $C^{1,1}$-spacetime that is globally hyperbolic. Let $h$ be a 
Riemannian metric on $M$ and let $K\comp M$. Then there exists some $C>0$ such that the $h$-length
of any causal curve taking values in $K$ is bounded by $C$.
\end{Lemma}
\begin{proof} 
It follows, e.g., from the proof of \cite[Lemma 14.13]{ON83} that $(M,g)$ is non-totally imprisoning,
i.e., there can be no inextendible causal curve that is entirely contained in $K$.
Now suppose that, contrary to the claim,  there 
exists a sequence $\sigma_k$ of causal curves valued in $K$ whose $h$-lengths tend 
to infinity. Parametrizing $\sigma_k$ by $h$-arclength we may assume that $\sigma_k: [0,a_k]\to K$,
where $a_k\to \infty$. Also, without loss of generality we may assume that $\sigma_k(0)$ converges
to some $q\in K$. 
Then by \cite[Th.\ 3.1(1)]{limit}\footnote{Note that the required result
remains
valid for $C^{1,1}$-metrics (in fact, even for continuous metrics): 
this follows exactly as in \cite[Th.\ 1.6]{CG}} 
one may extract a subsequence $\sigma_{k_j}$ that converges locally uniformly to an inextendible 
causal curve $\sigma$ in $K$,  
thereby obtaining a contradiction to non-total imprisonment.
\end{proof}

\begin{Lemma}\label{lengthcomp} 
Let $(M,g)$ be a globally hyperbolic $C^{1,1}$-spacetime and let
$g_\eps$ ($\eps>0$) be a net of smooth Lorentzian metrics %with $g_\eps\prec g$ for all $\eps$ 
such that $g_\eps$ converges locally uniformly to $g$ as $\eps\to 0$, and let $K\comp M$.
Then for each $\delta>0$ there exists some $\eps_0>0$
such that for each $\eps<\eps_0$ and each $g$-causal curve $\sigma$ taking values in $K$, the lengths
of $\sigma$ with respect to $g$ and $g_\eps$, respectively, satisfy:
\begin{equation}
L_{g}(\sigma) - \delta <  L_{g_\eps}(\sigma) < L_{g}(\sigma) + \delta.
\end{equation}
\end{Lemma}
\begin{proof} %Let $h$ be some background Riemannian metric as in Proposition \ref{CGapprox}. 
Since $g_\eps \to g$ uniformly on $K$, given any $\eta>0$ there exists some $\eps_0>0$ such that
for all $\eps<\eps_0$ and all $X\in TM|_K$ with $\|X\|_h=1$ we have
\begin{equation}
\|X\|_g -\eta \le \|X\|_{g_\eps} \le \|X\|_g + \eta.
\end{equation}
Consequently, for any $X\in TM|_K$ we have
\begin{equation}
\|X\|_g -\eta\|X\|_h \le \|X\|_{g_\eps} \le \|X\|_g + \eta\|X\|_h.
\end{equation}
Now if $\sigma:[a,b]\to K$ is any $g$-causal curve it follows that, for $\eps<\eps_0$,
\begin{equation}
\begin{split}
L_{g}(\sigma) -\eta L_h(\sigma) &= \int_a^b \|\sigma'(t)\|_g\,dt -\eta \int_a^b \|\sigma'(t)\|_h\,dt \le  \int_a^b \|\sigma'(t)\|_{g_\eps}\,dt
= L_{g_\eps}(\sigma)\\ 
&\le L_{g}(\sigma) +\eta L_h(\sigma).
\end{split}
\end{equation}
Finally, by Lemma \ref{lhbound} there exists some $C>0$ such that $L_h(\sigma)\le C$ for any $\sigma$ as above.
Hence, picking $\eta<\delta/C$ establishes the claim.
\end{proof}

\begin{Proposition}\label{prop2.2} Let $(M,g)$ be a future timelike-geodesically
complete $C^{1,1}$-spacetime. Let $S$ be a compact spacelike acausal
hypersurface in $M$, and 
let $p\in D^+(S)\setminus S$. Then 
\begin{itemize}
\item[(i)] 
$d_{\gec}(S,p) \to d(S,p) \quad (\eps\to 0)$.
\item[(ii)] There exists
a timelike geodesic $\gamma$ perpendicular to $S$ from $S$ to $p$ with $L(\gamma) = d(S,p)$.
\end{itemize}
\end{Proposition}
%%%%%%%%%%%%%%%%%%%%%%%%%%%%%%%%%%%%%%%%%%%%%%%%%%%%%%%%%%%%%%%%%%%%%%%%%%%%%%%%%%%%%%%%%%%%%%%%%%%%

%
Here we have dropped the subscript from the time separation function $d_g(S,p)$ 
and the length $L_g(\gamma)$ of the $C^{1,1}$-metric $g$ to simplify notations.
Also we remark that the proof of (i) below neither uses geodesic completeness
of $M$ nor compactness of $S$ and hence the $\gec$-distance converges even on
general $M$ for any closed spacelike acausal hypersurface $S$.

\begin{proof}
(i) Since $p\not\in S$ we have $c:=d(S,p)>0$. Let $0<\delta<c$. Then there exists a
$g$-causal curve $\alpha:[0,b]\to M$ from $S$ to $p$ with
$L_g(\alpha)>d(S,p)-\delta$. In particular, $\alpha$ is not a null curve, hence
there exist $t_1<t_2$ such that $\alpha|_{[t_1,t_2]}$ is nowhere null. In what
follows we adapt the argument from \cite[Lemma 2.4.14]{Chrusciel_causality} to
the present situation. Without loss of generality we may assume that $t_2=b$.
By Theorem~\ref{totally} we may find $0=s_0<s_1<\dots < s_N=b$ and totally
normal neighbourhoods $U_i$ ($1\le i \le N$) such that
$\alpha([s_i,s_{i+1}])\subseteq U_i$ for $0\le i <N$. 
By Proposition \ref{push-up1} we obtain that $\alpha(s_{N-1})\ll \alpha(b)$,
hence by Proposition \ref{longest}, the radial geodesic $\sigma_N$ from
$\alpha(s_{N-1})$ to $p$ is longer than $\alpha|_{[s_{N-1},b]}$, and it is timelike.
Next, we connect $\alpha(s_{N-2})$ via a timelike radial geodesic
$\sigma_{N-1}$ to some point on $\sigma_N$ that lies in $U_{N-1}$. 
Concatenating $\sigma_{N-1}$ with $\sigma_N$ gives a timelike curve longer
than $\alpha|_{[s_{N-2},b]}$. Iterating this
procedure we finally arrive at a timelike piecewise geodesic $\sigma$ from
$\alpha(0)=\sigma(0)\in S$ to $p$ of 
length $L_g(\sigma) \ge L_g(\alpha) > d(S,p) -\delta$. 

Since $L_{\gec}(\sigma) \to L_g(\sigma)$, we conclude that $L_\gec(\sigma) > d(S,p) -\delta$ for $\eps$ sufficiently small.
Moreover, $\sigma$ is $g$-timelike and piecewise $C^2$, hence is $\gec$-timelike for small $\eps$. Therefore, 
$d_\gec(S,p)\ge L_\gec(\sigma) > d(S,p) -\delta$ for $\eps$ small.

Conversely, if $\sigma$ is any $\gec$-causal curve from $S$ to $p$ then $\sigma$ is also
$g$-causal, hence lies entirely in the set $K:=J^-(p)\cap J^+(S,D(S))$. Since $D(S)$
is globally hyperbolic by Theorem \ref{A19} and Proposition \ref{dopen}, $K$ is compact by
Corollary \ref{baer}. %Pick $\eta>0$ such that $\eta d(S,p)<\delta$. 
Then by Lemma \ref{lengthcomp} (applied to the globally hyperbolic spacetime $(D(S),g)$),
for $\eps$ sufficiently small we have
\begin{equation}
L_{\gec}(\sigma) < L_g(\sigma) + \delta \le  d(S,p)+\delta.
\end{equation}
Consequently, $d_{\gec}(S,p) \le d(S,p) + \delta$
for $\eps$ sufficiently small. Together with the above this shows (i).\medskip

%Hence by Lemma \ref{lengthcomp}, there exists a net $c_\eps$ with
%$c_\eps\to 0$ as $\eps \to 0$ such that for any such curve $\sigma$ we have
%$$
%L_{\gec}(\sigma) \le (1+c_\eps)L_g(c) \le (1+c_\eps) d(S,p).
%$$
%It follows that $d_{\gec}(S,p) \le (1+c_\eps) d(S,p)$, so $d_{\gec}(S,p) \le d(S,p) + \delta$
%for $\eps$ sufficiently small. Together with the above this shows (i).\medskip

(ii) Since $\gec$ has narrower lightcones
than $g$, for each $\eps$ the point $p$ lies in $D^+_{\gec}(S)\setminus S$. 
Also, we may assume $\eps$ to be so small that $S$ is $\gec$-spacelike as well as 
$\gec$-acausal.
Then by smooth causality theory (e.g., \cite[Th.\ 14.44]{ON83}) there exists a
$\gec$-geodesic $\gamma_\eps$ that is $\gec$-perpendicular to $S$
and satisfies $L_{\gec}(\gamma_\eps) = d_{\gec}(S,p)$. Let $h$ be some
background Riemannian metric on $M$ and let $\gamma_\eps(0)=:q_\eps\in S$,
$\gamma_\eps'(0)=:v_\eps$. Without loss of generality we may suppose
$\|v_\eps\|_h=1$. Since $\{v\in TM\mid \pi(v)\in S, \ \|v\|_h=1\}$ is compact,
there exists a sequence $\eps_j\searrow 0$ such that $q_{\eps_j}\to q\in S$ and
$v_{\eps_j}\to v\in T_qM$. Denote by $\gamma_v$ the $g$-geodesic with
$\gamma(0)=q$, $\gamma'(0)=v$. 
To see that $\gamma$ is $g$-orthogonal to $S$, let $w\in T_qS$ and pick any
sequence $w_j\in T_{q_{\eps_j}}S$
converging to $w$. Then $g(v,w)=\lim \geck(v_{\eps_j},w_j)=0$. Consequently, $\gamma$ is
$g$-timelike.

Since $g$ is timelike geodesically complete,
$\gamma_v$ is defined on all of $\R$, so 
by standard ODE-results (cf., e.g., \cite[Sec.\ 2]{KSS}) for any $a>0$ there 
exists some $j_0$ such that for all $j\ge j_0$ the curve $\gamma_{\eps_j}$
is defined on $[0,a]$ and $\gamma_{\eps_j}\to \gamma$ in $C^1([0,a])$ (in fact, it follows directly
from this and the geodesic equation that this convergence even holds in $C^2([0,a])$).

For each $j$, let $t_j>0$ be such that $\gameck(t_j) = p$. Then by (i) we obtain
\begin{equation}
d(S,p) = \lim d_\geck(S,p) = \lim \int_0^{t_j} \|\gamma_{\eps_j}'(t)\|_\geck\,dt = \lim t_j \|v_{\eps_j}\|_\geck = \|v\|_g\lim t_j,
\end{equation}
so $t_j \to \frac{d(S,p)}{\|v\|_g}=:a$. Finally, for $j$ sufficiently large, all $\gameck$ are defined on 
$[0,2a]$ and we have $p=\gameck(t_j)\to \gamma(a)$, so $p=\gamma(a)$, as well as
\begin{equation}
\begin{split}
d(S,p) = \lim \int_0^{t_j} \|\gamma_{\eps_j}'(t)\|_\geck\,dt = \int_0^{a} \|\gamma'(t)\|_g\,dt = L_0^a(\gamma).
\end{split}
\end{equation}
\end{proof}

\section{Proof of the main result}\label{mainproof}
To prove Theorem \ref{hawking}, we first note that without loss of generality we may assume $S$ to be connected. 
Moreover, by Theorem \ref{covering} we may also assume $S$ to be achronal, and thereby acausal by
Lemma \ref{Sacausal} (replacing, if necessary, $M$ by a suitable Lorentzian covering space $\tilde M$ and 
$S$ by its isometric image $\tilde S$ in $\tilde M$). 
Note that since the light cones of $\gec$ approximate those of $g$ from the 
inside it follows that for $\eps$ small
$S$ is a spacelike acausal hypersurface with respect to $\gec$ as well.

We prove the theorem by contradiction and assume that $(M,g)$ is future timelike
geodesically complete.  Hence we may apply Proposition \ref{prop2.2} to obtain 
(using the notation from the proof of that result) for any $p\in D^+(S)\setminus S$:
\begin{enumerate}
 \item[(A)] $\exists\ g$-geodesic $\gamma$  $\perp_g$ $S$ realising the
time separation to $p$, i.e., $L(\gamma)=d(S,p)$.
 \item[(B)] $\exists\ \gec$-geodesics $\gamma_\eps$ $\perp_\gec$ $S$
realising the time separation to $p$, i.e., $L_{\gec}(\gamma_\eps)=d_{\gec}(S,p)$.
 \item[(C)] $\exists\ \eps_j\searrow 0$ such that $\gamma_{\eps_j}\to\gamma$ in
$C^1([0,a])$ for all $a>0$ (in fact, even in $C^2([0,a])$).
\end{enumerate}
We proceed in several steps.
\medskip

\noindent
{\bf Step 1.}
 $D^+(S)$ is relatively compact.\medskip
 
The future convergence of $S$ is given by $\conv =1/(n-1)\mathrm{tr}S_U$, with $S_U(V)=-\nabla_VU$
and $U$ the future pointing $g$-unit normal on $S$. Analogously, for each $\eps_j$ as in (C) we obtain the
future convergence $\conv_j$ of $S$ with respect $\geck$, and we denote the future-pointing $\geck$-unit normal 
to $S$ and the corresponding shape operator by $U_j$ and $S_{U_j}$, respectively. 
By Proposition \ref{CGapprox} (i), $\conv_j \to \conv$ uniformly on $S$.
Let $m:=\min_S \mathrm{tr}S_U =(n-1)\min_S \conv$, and $m_j:=
\min_S \mathrm{tr}S_{U_j} = (n-1)\min_S \conv_j$. By assumption, $m>0$, and by the above we obtain $m_j\to m$
as $j\to \infty$.

Let 
 \begin{equation}\label{6}
  b:=\frac{n-1}{m} 
 \end{equation}
and assume that there exists some $p\in D^+(S)\setminus S$ with $d(S,p)>b$. We will show that this
leads to a contradiction.

Since each $\gamma_{\eps_j}$ as in (C) is maximising until $p=\gamma_{\eps_j}(t_j)$, it
contains no $\geck$-focal point to $S$ before $t_j$. Setting $\tilde t_j:=(1-\frac{1}{j})t_j$
it follows that $\exp_\gec^\perp$ is non-singular on $[0,\tilde t_j]\gamma_{\eps_j}'(0) = [0,\tilde t_j]v_{\eps_j}$.
As this set is compact there
exist open neighbourhoods $W_j$ of $[0,\tilde t_j]v_{\eps_j}$ in the normal 
bundle $N_{\geck}(S)$ and $V_j$ of $\gamma_{\eps_j}([0,\tilde t_j])$ in $M$
such that $\exp_\geck^\perp:W_j\to V_j$ is a diffeomorphism. Due to $D_\geck(S)$ being open, 
we may also assume that $V_j\subseteq D_\geck(S)$. 

On $V_j$ we introduce the Lorentzian distance function
$r_j:=d_{\geck}(S,.)$ and set $X_j:=-\grad(r_j)$. 
Denote by $\tilde\gamma_j$ the re-parametrisation of $\gamma_{\eps_j}$
by $\geck$-arclength:
\begin{equation}
\tilde \gamma_j: [0,\tilde t_j \|v_{\eps_j}\|_{\geck}] \to M \quad           
\tilde \gamma_j(t) := \gamma_{\eps_j}(t/\|v_{\eps_j}\|_{\geck}).
\end{equation}
Then 
since $\tilde\gamma_j$ is maximising from $S$ to $p$ in $D^+_\geck(S)$, hence in particular in 
$V_j\cap J^+_\geck(S)$, it follows that $X_j(\tilde\gamma_j(t))=
{\tilde\gamma}'_j(t)$ for all $t\in [0, \tilde t_j \|v_{\eps_j}\|_{\geck}]$.
Next we define the shape operator corresponding to the distance function $r_j$ by
$S_{r_j}(Y):=\nabla^{\geck}_Y(\grad(r_j))$ for $Y\in \X(V_j)$. 
Then $S_{r_j}|_{S\cap V_j} = S_{U_j}|_{S\cap V_j}$ and the expansion
$\tilde\theta_j:=-\mathrm{tr} S_{r_j}$ satisfies the Raychaudhuri
equation (cf., e.g., \cite{Nat})
\begin{equation}
 X_j(\tilde\theta_j)+\mathrm{tr}(S_{r_j}^2)+\Ric_\geck(X_{j}, X_j)=0
\end{equation}
on $V_j$. Consequently, we obtain for
$\theta_j(t):=\tilde\theta_j\circ\tilde\gamma_j(t)$:
\begin{equation}
 \frac{d(\theta_j^{-1})}{dt}\geq\frac{1}{n-1}+\frac{1}{\theta_j^2}
 \Ric_\geck({\tilde\gamma}'_j,{\tilde\gamma}'_j).
\end{equation}
Now since by (C) the $\tilde\gamma_j$ converge in $C^1$ to the $g$-timelike geodesic
$\gamma$, it follows that there exist $\kappa<0$ and $C>0$ such that for all $j$
sufficiently large we have 
$g({\tilde\gamma}'_j(t),{\tilde\gamma}'_j(t))\le \kappa$ 
as well as $\|{\tilde\gamma}'_j(t)\|_h\le C$ for all $t\in [0, \tilde t_j \|v_{\eps_j}\|_{\geck}]$.

We are therefore in the position to apply Lemma \ref{(4)} to obtain that, 
for any $\delta>0$,
\begin{equation}\label{deltaest}
 \frac{d(\theta_j^{-1})}{dt}>\frac{1}{
n-1}-\frac{\delta}{\theta_j^2}
\end{equation}
for $j$ large enough.
Pick any $c$ with $b<c<d(S,p)$ and fix $\delta>0$ so small that
\begin{equation}\label{bc}
b < \frac{n-1}{\alpha m} <c,
\end{equation}
where $m$ is as in \eqref{6} and $\alpha:= 1 - (n-1)m^{-2}\delta$.
Analogously, let $\alpha_j:= 1 - (n-1)m_j^{-2}\delta$, so that $\alpha_j\to \alpha$ as
$j\to \infty$. 
Setting $d_j:=\tilde t_j \|v_{\eps_j}\|_{\geck}$, $\theta_j$ is defined
on $[0,d_j]$. Note that, for $j$ large, \eqref{bc} implies 
the right hand side of \eqref{deltaest} to be strictly positive at $t=0$.
Thus  $\theta_j^{-1}$ is initially strictly increasing and $\theta_j(0)<0$, so
\eqref{deltaest} entails that $\theta_j^{-1}(t)\in [-m_j^{-1},0)$
on its entire domain. From this we conclude that $\theta_j$ has no zero on 
$[0,d_j]$, i.e., that $\theta_j^{-1}$ exists on all of $[0,d_j]$. 
It then readily follows, again using \eqref{deltaest}, that $\theta_j^{-1}(t)
\ge f_j(t) := -m_j^{-1} + t \frac{\alpha_j}{n-1}$ 
on $[0,d_j]$. Hence $\theta_j^{-1}$ must go to zero before $f_j$ does,
i.e., $\theta_j^{-1}(t)\to 0$ as $t\nearrow T$ for some positive $T\le \frac{n-1}{\alpha_j m_j}$.

Here we note that due to $\lim d_j = \lim t_j \|v_{\eps_j}\|_{\geck} =d(S,p)$, 
for $j$ sufficiently large we have by \eqref{bc}
\begin{equation}
\frac{n-1}{\alpha_j m_j} < c < d_j.
\end{equation}
This, however, means that $\theta_j^{-1}\to 0$ within $[0,d_j]$, contradicting the
fact that $\theta_j$ is smooth, hence bounded, on this entire interval.

Together with (A) this implies that $D^+(S)$ is contained in the compact set $\beta(S\times[0,b])$
where 
\begin{equation}
 \beta:\ S\times[0,b]\to M,\quad (q,t)\mapsto\exp^g(t\,U(q)),
\end{equation}
Hence also the future Cauchy
horizon $H^+(S)= \overline{D^+(S)}\setminus I^-(D^+(S))$ is compact. 
\medskip

From here, employing the causality results developed in Appendix A,
we may conclude the proof exactly as in \cite[Th.\ 14.55B]{ON83}. For 
completeness, we give the full argument.\medskip

\noindent
{\bf Step 2.} The future Cauchy horizon of $S$ is nonempty.\medskip

Assume to the contrary that $H^+(S)=\emptyset$. Then 
$I^+(S)\subseteq D^+(S)$: for $p\in S$, a future-directed timelike curve
$\gamma$ starting at $p$ lies
initially in $D^+(S)$ (using Proposition \ref{dopen}, or Lemma \ref{locallycauchy}). 
Hence if $\gamma$ leaves $D^+(S)$, it must meet $\partial D^+(S)$ and by Lemma \ref{Dboundary}
it also meets $H^+(S)$ (since $S$ is achronal it can't intersect $S$ again). But then
$H^+(S)$ wouldn't be empty, contrary to our assumption. Hence $I^+(S)\subseteq D^+(S)$. By Step 1,
then, $I^+(S)\subseteq \{p\in M \ |\ d(S,p)\leq b \}$ and hence $L(\gamma)\leq b$
for any timelike future-directed curve emanating from $S$, which is a
contradiction to timelike geodesic completeness of $M$.
\medskip

\noindent
{\bf Step 3.} The following extension of (A) holds:

\begin{enumerate}
 \item [(A')] $\forall\ q\in H^+(S)$ $\exists\ g$-geodesic $\gamma$  $\perp_g$
$S$ realising the time separation and $L(\gamma)=d(S,q)\leq b$.
\end{enumerate}

Consider the set $B\subseteq N(S)$ consisting of the zero
section and all future pointing causal vectors $v$ with
$\|v\|\leq b$. $B$ is compact by the compactness of $S$.

By definition there is a sequence $q_k$ in $D^+(S)$ that converges to $q$.
For any $q_k$ there is a geodesic as in (A) and hence a vector $v_k\in B$
with $\exp(v_k)=q_k$. By the compactness of $B$ we may assume that $v_k\to v$
for some $v\in B$ and hence by continuity $q_k\to\exp(v)$. Moreover, we have
by construction that $\|v_k\|=d(S,q_k)$. Since $d$ is lower semicontinuous (Lemma
\ref{lsc}), $\|v\|\geq d(S,q)$.

As $\gamma_v$ is perpendicular to $S$, hence timelike, our completeness assumption 
implies that it is defined on $[0,1]$. Thus it
runs from $S$ to $q$ and has length $\|v\|$, which implies $d(S,q)=\|v\|\leq b$.

\medskip
{\bf Step 4.} The map $p\mapsto d(S,p)$ is strictly decreasing along past
pointing generators of $H^+(S)$.\medskip

By Proposition \ref{horizon} (iii), $H^+(S)$ is generated by past-pointing 
inextendible null geodesics. Suppose that 
$\alpha: I\to M$ is such a generator, and let $s$, $t\in I$, $s<t$.  
Using (A')  we obtain a past pointing timelike geodesic $\gamma$ from
$\alpha(t)$ to $\gamma(0)\in S$ of length $d(S,\alpha(t))$. Then 
arguing as in the proof of Proposition \ref{prop2.2} (i) we may construct
a timelike curve $\sigma$ from $\alpha(s)$ to $\gamma(0)$ that is strictly longer
than the concatenation of $\alpha|_{[s,t]}$ and $\gamma$. Therefore,
\begin{equation}
 d(S,\alpha(s))\geq
L(\sigma)>L(\alpha|_{[s,t]}+\gamma) = L(\gamma)=d(S,\alpha(t)).
\end{equation}
{\bf Step 5.} $(M,g)$ is \emph{not} future timelike geodesically complete.\medskip

By step 1, $H^+(S)$ is compact and by Lemma \ref{lsc} $p\mapsto d(S,p)$
is lower semicontinuous, hence attains a finite minimum at
some point $q$ in $H^+(S)$. But then taking a past pointing generator
of $H^+(S)$ emanating from $q$ according to Proposition
\ref{horizon} (iii) gives a contradiction to step 4.
\hfill$\Box$

\begin{appendix}
\section*{Appendix A: Results from $\mathbf C^{1,1}$-causality theory}
\setcounter{section}{1}
\def\thesection{\Alph{section}}

In this appendix we collect those results on the causality of $C^{1,1}$-metrics
that are used in the main text, that is, \ref{totally}, \ref{longest}, \ref{push-up1},
 \ref{Dboundary}, \ref{lsc}, \ref{A19}--\ref{locallycauchy},
\ref{baer}, \ref{Sacausal},  \ref{covering}, as well as
those supplementary statements that are used to prove these, or to secure the compatibility 
with \cite{HE} as explained in Section \ref{prelim} (\ref{lipdevelopment} 
and \ref{cauchysurface}). Using the results on basic causality theory
of $C^{1,1}$-metrics established in \cite{CG,M,KSS,KSSV}, see Theorem \ref{lcb} to Lemma \ref{IJlemma}
below, combined with the standard proofs in the smooth case, it is a routine matter to prove the 
remaining results. So instead of providing
full proofs we accurately collect all facts and previous statements entering the
respective proofs. In this way we provide a concise chain of arguments
on the one hand establishing the results and on the other hand showing at which
places regularity issues have to be taken into account. Our
presentation is essentially based on the one of \cite{ON83}.

We first recall a few fundamental results from $C^{1,1}$-causality theory that
are used throughout the proofs of this section. From now (unless explicitly stated otherwise)
we will exclusively work on a $C^{1,1}$-spacetime $(M,g)$. 
Denoting by
$\tilde{Q}: T_pM\rightarrow \R$, $v\mapsto g_p(v,v)$ the quadratic form on the tangent
space of a Lorentzian manifold, we have:

\begin{Theorem}\label{lcb}
Let $(M,g)$ be a $C^{1,1}$-spacetime, and let $p\in M$. Then $p$ has a basis of normal neighbourhoods $U$, 
$\exp_p: \tilde U\to U$ a bi-Lipschitz homeomorphism, such that:
 \begin{equation*}
 \begin{split}
   I^{+}(p,U)=\exp_{p}(I^{+}(0)\cap \tilde{U})\\
   J^{+}(p,U)=\exp_{p}(J^{+}(0)\cap \tilde{U}) \\
   \partial I^{+}(p,U) = \partial J^{+}(p,U) =\exp_{p}(\partial I^{+}(0)\cap \tilde{U})
 \end{split}  
 \end{equation*}
Here, $I^+(0)=  \{v\in T_pM \mid \tilde Q(v)<0 \}$, and $J^+(0)= \{v\in T_pM \mid \tilde Q(v)\le  0 \}$.
In particular, $I^+(p,U)$ (respectively $J^+(p,U)$) is open (respectively closed) in $U$.
\end{Theorem}
For a proof, see \cite[Th.\ 1.23]{M} or \cite[Th.\ 3.9]{KSSV}.
\begin{Corollary}\label{lipisc1} 
  Let $U\subseteq M$ be open, $p\in U$. Then the sets $I^+(p,U)$,
  $J^+(p,U)$ remain unchanged if
  Lipschitz curves are replaced by piecewise $C^1$ curves, or in fact
  by broken geodesics.
\end{Corollary}
See \cite[Th.\ 1.27]{M} or \cite[Cor.\ 3.10]{KSSV}.

\medskip\noindent
The usual convexity properties also hold for $C^{1,1}$-metrics: If $U$ is a
normal neighbourhood of each of its points then it is called totally
normal or (geodesically) convex. Any pair of its points can then be connected by
a unique geodesic contained in $U$. The following result (\cite[Th.\ 4.1]{KSS}, 
\cite[Th. 1.16]{M}) guarantees existence of such neighbourhoods:

\begin{Theorem} \label{totally} Let $M$ be a smooth manifold with a
  $C^{1,1}$-pseudo-Riemannian metric $g$. Then each point $p\in M$
  possesses a basis of totally normal neighbourhoods.
\end{Theorem}
%For the proof, see \cite[Th.\ 4.1]{KSS}, \cite[Th.\ 1.16]{M}.
Concerning curve-lengths in normal neighbourhoods, \cite[Th.\ 1.23]{M} gives:
\begin{Proposition}\label{longest}
Let $U$ be a normal neighborhood of $p\in M$. If $p\ll q$ for a point $q\in U$,
then the radial geodesic segment $\sigma$ is the unique longest timelike curve in 
$U$ connecting $p$ and $q$.
\end{Proposition}

\medskip\noindent
The following result provides more information about causal curves intersecting
the boundary of $J^+(p,U)$:
%%%%%%%%%%%%%%%%%%%%%%%%%%%%%%%%%%%%%%%%%%%%%%%%%%%%%%%%%%%%%%%%%%%%%%%%%%%%%%%%%%%%%%%%%%%%%%%%%%%%%%%%%
%%%%%%%%%%%%%%%%%%%%%%%%%%%%%%%%%%%%%%%%%%%%%%%%%%%%%%%%%%%%%%%%%%%%%%%%%%%%%%%%%%%%%%%%%%%%%%%%%%%%%%%%%
\begin{Corollary}\label{boundary} 
  Let $U$ be as in Th.\ \ref{lcb}, suppose that $\alpha: [0,1]
  \to U$ is causal and $\alpha(1)\in \partial J^+(p,U)$. Then $\alpha$
  lies entirely in $\partial J^+(p,U)$ and there exists a
  reparametrisation of $\alpha$ as a null-geodesic segment.
\end{Corollary}
See \cite[Th. 1.23]{M} or \cite[Cor.\ 3.11]{KSSV}.

The following fundamental push-up principle (\cite[Lemma 1.22]{CG}) in fact even holds 
for Lipschitz (or, more generally, causally plain continuous) metrics:
\begin{Proposition}\label{push-up1} 
Let $g$ be a $C^{0,1}$-metric on $M$ and let $p,\, q,\, r\in M$ with $p\le q$ and 
$q\ll r$ or $p\ll q$ and $q\le r$.  Then $p\ll r$.
\end{Proposition}

\begin{Proposition}\label{IJprop} Let $U\subseteq M$ as in Th.\ \ref{lcb} be totally normal. 
\begin{itemize}
\item[(i)] Let $p$, $q\in U$. Then $q\in I^+(p,U)$ (resp.\ $\in J^+(p,U)$) if and only
if $\overrightarrow{pq}:=\exp_p^{-1}(q)$ is future-directed timelike (resp.\ causal).
Also, $(p,q)\mapsto  \overrightarrow{pq}$ is continuous.
\item[(ii)] $J^+(p,U)$ is the closure of $I^+(p,U)$ relative to $U$.
\item[(iii)] The relation $\le$ is closed in $U\times U$.
\item[(iv)] If $K$ is a compact subset of $U$ and $\alpha: [0,b)\to K$ is causal, then $\alpha$
can be continuously extended to $[0,b]$. 
\end{itemize}
\end{Proposition}
For a proof, see \cite[Prop.\ 3.15]{KSSV}.

\begin{Lemma}\label{IJlemma}
The relation $\ll$ is open. Moreover, for $A\subseteq U \subseteq M$, where $U$ is open, we have:
\begin{equation}
\begin{split}
I^+(A,U) &= I^+(I^+(A,U)) = I^+(J^+(A,U)) = J^+(I^+(A,U))\\ 
&\subseteq J^+(J^+(A,U)) = J^+(A,U)
\end{split}
\end{equation}
\end{Lemma}
See \cite[Cor.\ 3.12, Cor.\ 3.13]{KSSV}.

\begin{Lemma}\label{easylemma}
Let $S\subseteq M$ be achronal. Then:
\begin{enumerate}
 \item[(i)] $S\subseteq D^{\pm}(S)\subseteq S\cup I^{\pm}(S)$
 \item[(ii)] $D^+(S)\cap I^-(S)=\emptyset$
 \item[(iii)] $D^+(S)\cap D^-(S)=S$
 \item[(iv)] $D(S)\cap I^{\pm}(S)=D^{\pm}(S)\setminus S$.
\end{enumerate}
\end{Lemma}

\noindent
As in the smooth case, these properties are immediate from the definitions.
\begin{Lemma}\label{pastinex}
Let $S$ be a closed set and let $\gamma$ be a past inextendible causal curve
starting at $p$ that does not meet $S$. Then:
\begin{enumerate}
 \item[(i)] For any $q\in I^+(p,M\setminus S)$ there exists a past inextendible timelike
 piecewise geodesic $\tilde{\gamma}$ starting at $q$ that does not meet $S$;
 \item[(ii)] If $\gamma$ is not a null geodesic, there exists
 a past inextendible timelike piecewise geodesic $\tilde{\gamma}$ starting at $p$ that
 does not meet $S$.
\end{enumerate}
\end{Lemma}
The proof of the first statement carries over from the smooth case, see \cite[Lemma\ 14.30]{ON83},
using Proposition \ref{push-up1}. For the second
statement (to avoid the variational calculus-based proof of \cite[Lemma\ 14.30]{ON83}) 
we need the following argument:
\begin{Lemma}
Let $S$ be a closed set and let $\alpha: [0,\infty)\rightarrow M\setminus S$ be 
a past directed causal curve which is not a null geodesic. Then there exists $a>0$
such that $\alpha(a)\ll \alpha(0)$ (with $\ll$ the relation on $M\setminus S$).
\end{Lemma}
%\begin{remark}
%Note that here by .  
%\end{remark}
\begin{proof}
Suppose to the contrary that there is no point on the curve $\alpha$ which can be
timelike related to $\alpha(0)$ within $M\setminus S$. Using Theorem \ref{totally} 
we can cover $\alpha$ by totally normal neighbourhoods 
$U_i$ with $U_i\subseteq M\setminus S$ since $M\setminus S$ is open. Let $t_0=0<t_1<t_2\ ...$
such that $\alpha|_{[t_i,t_{i+1}]}\subseteq U_{i+1}$. By our assumption, it follows that $\alpha|_{[t_0,t_1]}$
lies in $\partial J^-(\alpha(0),U_1)$. Hence, by Corollary \ref{boundary}, $\alpha|_{[t_0,t_1]}$
is a null geodesic. Iterating this
procedure we obtain that $\alpha$ is a null geodesic, a contradiction.
%Hence there exists a timelike curve $\tilde{\alpha}$ from $\alpha(0)$ to $\alpha(a)$, 
%for some $a>0$, which is entirely in $M\setminus S$. 
\end{proof}
Using this, the proof of Lemma \ref{pastinex} can be concluded as in \cite[Lemma\ 14.30]{ON83}.

\begin{Lemma}\label{lipdevelopment}
Let $S$ be a closed achronal hypersurface. Then the Cauchy development defined
with Lipschitz curves, $D^{+}(S)$, coincides with the one defined with piecewise
$C^1$-curves, $D^{+}_{C^1}(S)$.
\end{Lemma}
\begin{proof}
Obviously, $D^{+}(S) \subseteq D^{+}_{C^1}(S)$. Now suppose there existed some $p\in D^{+}_{C^1}(S)\setminus D^{+}(S)$.
Then there would exist a past inextendible Lipschitz causal curve $\gamma$ from $p$ such that
$\gamma \cap S=\emptyset$. By Theorem \ref{totally}, we may cover $\gamma$ by totally normal neighbourhoods
$U_1,...,U_N,...$ such that $\gamma([s_i,s_{i+1}])\subseteq U_{i+1}$, $\forall i$. 
Then we distinguish two cases: If 
$\gamma([s_i,s_{i+1}])\subseteq \partial J^+(\gamma(s_i),U_i)$ for all $i$, 
then by Corollary \ref{boundary}
$\gamma$ is a piecewise null geodesic and therefore piecewise $C^{1}$, a contradiction.
The second possibility is that $\exists i,\ \exists t\in (s_i, s_{i+1})$ such that $\gamma(s_i)\ll \gamma(t)$.
But then Lemma \ref{pastinex} (ii) gives a contradiction.
\end{proof}

\begin{Lemma}\label{barD}
Let $S$ be a closed achronal set. Then $\overline{D^+(S)}$ is the set of all points
$p$ such that every past inextendible timelike curve through $p$ meets $S$.
\end{Lemma}
This can be shown as in \cite[Lemma\ 14.51]{ON83}, using Theorem \ref{lcb},
Theorem \ref{totally}, Lemma 
\ref{easylemma} (i), and Lemma \ref{pastinex} (i).

\begin{Lemma}\label{Dboundary}
Let $S$ be a closed achronal set. Then $\partial D^{\pm}(S)=S\cup H^{\pm}(S)$.
\end{Lemma}
For a proof, follow that of \cite[Lemma\ 14.52]{ON83}, using Lemma \ref{easylemma} (i),
Theorem \ref{lcb}, Proposition \ref{push-up1} and Lemma \ref{barD}.

\begin{Lemma}\label{distancefnc} \
\begin{enumerate}
 \item[(i)] $d(p,q)>0$ if and only if $p\ll q$
 \item[(ii)] If $p\leq q \leq r$, then $d(p,q)+d(q,r)\leq d(p,q)$.
\end{enumerate}
 \end{Lemma}
Using Proposition \ref{push-up1}, this follows as in \cite[Lemma\ 14.16]{ON83}. 
 
\begin{Lemma}\label{lsc}
$d$ is lower semi-continuous. 
\end{Lemma}
This can be proved following \cite[Lemma\ 14.17]{ON83}, 
using Lemma \ref{distancefnc} and Theorem \ref{lcb}.

\begin{Lemma}\label{edge}
Let $S\subseteq M$ be achronal. Then $\bar{S}\setminus S\sse \edge(S)$.
\end{Lemma}
The proof can be carried out as indicated in the proof of \cite[Cor.\ 14.26]{ON83},
using Theorem \ref{lcb}, and the fact that the closure of any achronal set $S$ is achronal, 
which follows from Lemma \ref{IJlemma}.

\begin{Proposition}\label{topedge}
An achronal set $S$ is a topological hypersurface if and only if $S$ contains
no edge points.
\end{Proposition}
For a proof, see \cite[Prop.\ 14.25]{ON83}, employing Theorem \ref{lcb}.

\begin{Corollary}\label{emptyedge}
An achronal set $S$ is a closed topological hypersurface if and only if $\edge(S)$
is empty.
\end{Corollary}
This can be seen as in \cite[Cor.\ 14.26]{ON83}, using
Lemma \ref{edge} and Proposition \ref{topedge}.

\begin{Lemma}\label{causalcauchyhyp}
Let $S\subseteq M$ be a Cauchy hypersurface. Then:
\begin{enumerate}
 \item[(i)] $S$ is a closed achronal topological hypersurface.
 \item[(ii)] Every inextendible causal curve intersects $S$.
\end{enumerate}
\end{Lemma}
To show this one may follow \cite[Lemma 14.29]{ON83}, using
Lemma \ref{IJlemma} as well as Corollary \ref{emptyedge}, and Lemma \ref{pastinex} (i)
(replacing \cite[Cor.\ 14.27]{ON83}).

\begin{Lemma}\label{I-+}
Let $S$ be an achronal set and let $p\in D(S)^\circ$. Then every inextendible
causal curve through $p$ meets both $I^-(S)$ and $I^+(S)$.
\end{Lemma}
The proof carries over from  \cite[Lemma 14.37]{ON83}, and uses
Theorem \ref{lcb}, Lemma \ref{easylemma}, as well as (the proof of) Lemma \ref{pastinex} (i).

\begin{Theorem}\label{A19}
Let $S$ be achronal. Then $D(S)^\circ$ is globally hyperbolic.
\end{Theorem}
The proof can be done following \cite[Th.\ 14.38]{ON83}. The constructions 
used there (limit sequences of causal curves and their
properties, existence of convex refinements of open coverings) 
all carry over to the $C^{1,1}$-setting, using 
Theorems \ref{lcb}, \ref{totally}, Propositions \ref{push-up1}, \ref{IJprop}, and 
Lemmas \ref{easylemma}, \ref{I-+}.

\begin{Proposition}\label{dopen}
Let $S$ be a closed acausal topological hypersurface. Then $D(S)$ is open.
\end{Proposition}
This proposition can be proved following \cite[Lemma\ 14.43]{ON83}, using
Theorem \ref{lcb}, Lemma \ref{easylemma},
Proposition \ref{push-up1}, Proposition \ref{IJprop} and Proposition \ref{topedge}.

\begin{Proposition}\label{horizon}
Let $S$ be a closed acausal topological hypersurface. Then
\begin{enumerate}
 \item[(i)] $H^+(S)=I^+(S)\cap \partial D^+(S)=\overline{D^+(S)}\setminus D^+(S)$;
 \item[(ii)] $H^+(S)\cap S=\emptyset$
 \item[(iii)] $H^+(S)$ is generated by past inextendible null geodesics that are 
 entirely contained in $H^+(S)$.
\end{enumerate}
\end{Proposition}
The proof can be done combining \cite[Prop.\ 14.53]{ON83} and Theorem \ref{lcb}, 
Lemmas \ref{easylemma}, \ref{pastinex} (ii),
\ref{barD}, \ref{Dboundary}, and Proposition \ref{dopen}.

\begin{Lemma}\label{locallycauchy}
Let $S$ be a spacelike hypersurface and let $p\in S$. Then there exists a neighborhood
$V$ of $p$ such that $V\cap S$ is a Cauchy hypersurface in $V$.
\end{Lemma}
\begin{proof} 
Let $\hat{g}_{\eps}$ be smooth metrics approximating $g$ from the outside 
as in Prop.\ \ref{CGapprox}. Then given any compact neighborhood $W$ of $p$ in $M$
there exists some $\eps>0$ such that $W\cap S$ is spacelike for $\hat g_\eps$.
From the smooth theory (e.g., \cite[Lemma A.5.6]{Baer}) we obtain that
there exists a neighborhood $V\subseteq W$ such that $V\cap S$ is a 
Cauchy hypersurface in $V$ for $\hat{g}_\eps$, and consequently also for $g$. 
\end{proof}

\begin{Lemma}\label{pastcompact}
Let $S$ be an achronal set in $M$ and let $p\in D(S)^\circ \setminus I^-(S)$. Then
$J^-(p)\cap D^+(S)$ is compact.
\end{Lemma}
The proof follows \cite[Th.\ 14.40]{ON83}, using  Theorem \ref{lcb}, Proposition
\ref{push-up1}, Proposition \ref{IJprop}, Lemma \ref{easylemma}, and Lemma \ref{I-+}.

\begin{Lemma}\label{relcompact}
Let $K$ be a compact subset of $M$ and let $A\sse M$ be such that, $\forall p\in M$,
$A\cap J^{+}(p)$, respectively $A\cap J^{-}(p)$, is relatively compact in $M$. Then 
$A\cap J^{+}(K)$, respectively $A\cap J^{-}(K)$, is relatively compact in $M$.
\end{Lemma}
The proof can be carried out as in \cite[Lemma A.5.3]{Baer}, based on Theorem \ref{lcb}.

\begin{Proposition}\label{fum}
 Let $U\subseteq M$ be open and globally hyperbolic. Then the causality relation $\leq$ of $M$ is
 closed on $U$.
\end{Proposition}
This can be proved as in \cite[Lemma 14.22]{ON83} (based on \cite [Lemma 14.14]{ON83}),
using Theorem \ref{totally} and  Propositions \ref{longest}, \ref{IJprop}.

\begin{Corollary}\label{baer}
Let $S$ be a Cauchy hypersurface in a globally hyperbolic manifold $M$ and let $K$ be
compact in $M$. Then $S\cap J^{\pm}(K)$ and $J^{\mp}(S)\cap J^{\pm}(K)$ are compact.
\end{Corollary}
This follows as in \cite[Lemma A.5.4]{Baer}, using Proposition
\ref{IJprop}, Lemmas \ref{pastcompact}, \ref{relcompact} and Proposition \ref{fum}.

We give a proof of the following result, again to avoid the variational calculus-based 
argument in \cite[Lemma 14.42]{ON83}.
\begin{Lemma}\label{Sacausal}
Any achronal spacelike hypersurface $S$ is acausal. 
\end{Lemma}
\begin{proof}
Let $\alpha: [0,1]\rightarrow M$ be a future directed causal curve with endpoints 
$\alpha(0)$ and $\alpha(1)$ in $S$. If $\alpha$ is not a null-geodesic, by Proposition 
\ref{push-up1}, we can connect $\alpha(0)$ with $\alpha(1)$ also by a timelike curve, which is a 
contradiction to the achronality of $S$. Now let $\alpha$ be a null geodesic. By Lemma 
\ref{locallycauchy}, there exists a neighborhood $U$ around $\alpha(0)$ in which $S\cap U$ is 
a Cauchy hypersurface. Since $\alpha$ is $C^2$ and causal, it must be transversal to S, so 
it contains points in $J^+(S,U)\setminus S$. Then we can connect any such point with some 
point in $S\cap U$ by a timelike curve within $U$. Concatenating this curve with the remainder of 
$\alpha$, we obtain a curve that is not entirely null and meets S twice. As above, this
gives a contradiction to achronality.
\end{proof}

\begin{Proposition}\label{cauchysurface}
Let $S$ be a spacelike hypersurface in $M$. Then $S$ is a Cauchy hypersurface if and only 
if every inextendible causal curve intersects $S$ precisely once.
\end{Proposition}
\begin{proof}
Let $S$ be a Cauchy hypersurface and let $\alpha$ be an inextendible causal curve. By Lemmas 
\ref{causalcauchyhyp}(i), \ref{Sacausal}, $\alpha$ intersects $S$ at most once. Also, 
by Lemma \ref{causalcauchyhyp} (ii), it has to intersect $S$ at least once, hence the result.
\end{proof}
\medskip

The remaining statements in this appendix serve to justify that in the proof of the
main result in Section \ref{mainproof} we may without loss of generality assume 
$S$ to be achronal. This is done using a covering argument, as in \cite{HE,ON83}.
A key ingredient in adapting this construction to the $C^{1,1}$-setting
is the following consequence of \cite[Th.\ 1.39]{M}:

\begin{Theorem}\label{subnormalnbhds}
Let $M$ be a smooth manifold with a $C^{1,1}$-Lorentzian metric and 
let $S$ be a semi-Riemannian submanifold of $M$. Then the normal bundle $N(S)$ is Lipschitz.
Moreover, there exist neighbourhoods $U$ of the zero section in $N(S)$ and 
$V$ of $S$ in $M$ such that
\begin{equation*}
\exp^\perp : U\rightarrow V
\end{equation*}
is a bi-Lipschitz homeomorphism.
\end{Theorem}

\begin{Lemma}\label{Sachronal}
Let $S$ be a connected closed spacelike hypersurface in $M$.
\begin{enumerate}
 \item[(i)] If the homomorphism of fundamental groups $i_{\sharp}: \pi_1(S)\rightarrow \pi_1(M)$ induced 
 by the inclusion map $i: S\hookrightarrow M$ is onto, then $S$ separates $M$ (i.e., $M\setminus S$ is
 not connected).
 \item[(ii)] If $S$ separates $M$, then $S$ is achronal.
\end{enumerate}
\end{Lemma}
The proof carries over from \cite[Lemma\ 14.45]{ON83} using Theorem \ref{subnormalnbhds},
Theorem \ref{lcb} and a result from intersection theory,
namely, that a closed curve which intersects a closed hypersurface $S$ precisely
once and there transversally, is not freely homotopic to a closed curve which does
not intersect $S$, cf. \cite[p.\ 78]{GP}. The only change to \cite[Lemma\ 14.45]{ON83}
is that for the curve $\sigma$ we take a geodesic, which automatically is a $C^1$-curve
(in fact, even $C^2$), so that the intersection theory argument applies. 

\begin{Theorem}\label{covering}
Let $S$ be a closed, connected, spacelike hypersurface in $M$.
Then there exists a Lorentzian covering $\rho: \tilde{M}\rightarrow M$ and 
an achronal closed spacelike hypersurface $\tilde{S}$ in $\tilde{M}$ which 
is isometric under $\rho$ to $S$.
\end{Theorem}
The proof carries over from \cite[Prop.\ 14.48]{ON83} using Lemma \ref{Sachronal}.
\end{appendix}

%\newpage

\medskip\noindent
{\bf Acknowledgements.}  We would like to thank James D.\ E.\ Grant and Ettore Minguzzi for
helpful discussions.  The authors acknowledge the support of FWF
projects P23714 and P25326, as well as OeAD project WTZ CZ 15/2013.

\end{document}